\documentclass[runningheads]{llncs}
\usepackage{amssymb}
\usepackage[ruled,noend]{algorithm2e}
\usepackage{graphicx}
\usepackage{subfigure}
\usepackage[strings]{underscore}
\usepackage{import}
\usepackage{xifthen}
\usepackage{pdfpages}
\usepackage{geometry}
\usepackage{transparent}

\usepackage{amsthm}
\usepackage{amsmath}
\theoremstyle{definition}
\newtheorem{observation}[]{Observation}

\begin{document}
\title{Broadcast Graph Is NP-complete\thanks{Our work was supported in part by the Guangdong Provincial Key Laboratory IRADS (2022B1212010006, R0400001-22)  and in part by Guangdong Higher Education Upgrading Plan (2021-2025) with UIC research grant UICR0400025-21.}}

\author{Jinghan Xu\inst{1,2}\orcidID{0009-0009-7861-7784} \and \\Zhiyuan Li\inst{1,3}\orcidID{0000-0002-1991-2603}}
\authorrunning{J.Xu and Z.Li}
\institute{Guangdong Provincial Key Laboratory IRADS,\\ BNU-HKBU United International College, Zhuhai 519087, China \and \email{xujinghan@uic.edu.cn}
\and \email{goliathli@uic.edu.cn}}

\maketitle				  
\begin{abstract}
The broadcast model is widely used to describe the process of information dissemination from a single node to all nodes within an interconnected network.
 In this model, a graph represents the network, where vertices correspond to nodes and edges to communication links. 
  The efficiency of this broadcasting process is evaluated by the broadcast time, the minimum discrete time units required to broadcast from a given vertex. Determining the broadcast time is referred to as the problem \textsc{Broadcast Time}.
   The set of vertices with the minimum broadcast time among the graph is called the broadcast center.  Identifying this center or determining its size are both proven to be NP-hard.
    For a graph with $n$ vertices, the minimum broadcast time is at least $\lceil \log_2 n\rceil$. The \textsc{Broadcast Graph} problem asks in a graph of $n$ vertices, whether the broadcast time from any vertex equals $\lceil \log_2 n\rceil$. 
  Extensive research over the past 50 years has focused on constructing broadcast graphs, which are optimal network topologies for one-to-all  communication efficiency. However, the computational complexity of the \textsc{Broadcast Graph} problem has rarely been the subject of study. We believe that the difficulty lies in the mapping reduction for an NP-completeness proof. Consequently, we must construct broadcast graphs for yes-instances and non-broadcast graphs for no-instances. The most closely related result is the NP-completeness of \textsc{Broadcast Time} proved by Slater et al. in 1981. More recently, Fomin et al. has proved that \textsc{Broadcast Time} is fixed-parameter tractable.
   In this paper, we prove that \textsc{Broadcast Graph} is NP-complete by proving a reduction from \textsc{Broadcast Time}. We also improve the results on the complexity of the broadcast center problem. We show \textsc{Broadcast Center Size} is in $\Delta_p^2$, and is $D^P$-hard, implying a complexity upper bound of $\Delta_p^2$-complete and a lower bound of $D^P$-hard (NP $\subseteq D^P\subseteq \Delta_p^2$).
\end{abstract}
\section{Introduction}
In distributed file systems, in order to maintain data consistency, modified files from one node are synchronized with their copies stored in all other nodes in the network. The synchronization is a one-to-all dissemination procedure, called $broadcast$.
 The broadcast process starts with one node named $originator$ holding the information, and each informed node making parallel calls to adjacent nodes. A $call$ refers to the transmission of information from one $sender$ to one adjacent $receiver$ in one discrete time unit. This paper is based on the classic broadcast model with the following assumptions. 
\begin{itemize}
\item A broadcast is divided into discrete time units. 
\item Only one node, the originator, has the information at time unit 0.
\item Each informed node can inform at most one uninformed neighbor per time unit.
\item The broadcast ends when all the nodes in the network are informed.  
\end{itemize}
A network is modeled as a connected graph $G\ =\ (V,E)$, where $V$ is the vertex set representing the nodes, and $E$ is the edge set representing the communication links.

Other dissemination models, such as unicast (one-to-one), multicast (one-to-many), and gossip (all-to-all), are also discussed in the literature \cite{fertin_odd_2017,boyd_randomized_2006,hedetniemi_survey_1988}.

A call through an edge is represented by a directed edge from the sender (an informed vertex) to the receiver (an uninformed vertex). The complete set of calls in the graph forms a $broadcast\ tree$ of $G$, a directed spanning tree rooted at the originator. The tree visually represents the $broadcast\ scheme$, sequencing calls over time units to illustrate the entire broadcast process. In a broadcast scheme, if an informed vertex $v$ does not make a call at time $t$, then $v$ is $idle$ at that time unit; otherwise, $v$ is $busy$.

The distance between 2 vertices $u$ and  $v$, denoted by $d(u,v)$, is the number of edges in the shortest path between them. During the broadcast process, the information passing from an informed vertex $u$ to a vertex $v$ costs at least $d(u,v)$ time units.

The $broadcast\ time$, $b(G,v)$, denotes the minimum time units required to broadcast from vertex $v$ in graph $G$. The broadcast time of the graph, $b(G)$, is the maximum broadcast time from any originator $u\in V$, formally $b(G) = \mathop{\max}\limits_{u\in V}(b(G,u))$. 

The problem of finding the broadcast time $b(G,v)$ is NP-complete \cite{slater_information_1981}. More recently, it has been shown to be fixed-parameter tractable (FPT) using three different kernelizations: the vertex cover, the feedback edge set, or the input deadline of broadcast time \cite{FOMIN2024114508}. A follow-up study on parameterized complexity of broadcast time \cite{tale2024double} indicates that the problem remains NP-complete with a feedback vertex set of size 1 (and therefore of tree width 2). Concurrently, a double exponential lower bound is established when parameterized by the solution size.

In the classic broadcast model, each sender can inform at most one receiver per time unit. Thus, the number of informed vertices can at most double in one time unit. For arbitrary graph $G$ on $n$ vertices, $b(G)\geq \lceil \log n\rceil$. Throughout this paper, the base of the logarithm is assumed to be 2 and is therefore omitted.

A graph $G$ on $n$ vertices is called a $broadcast\ graph$ if $b(G) = \lceil \log n\rceil$, the theoretically lowest broadcast time. The value of $n$ is always partitioned into discrete ranges between two consecutive power of 2, $2^{t-1}+1\leq n\leq 2^t$, because broadcast graphs within this range have the same broadcast time $t$. The broadcast graph structure ensures the most efficient one-to-all data transmission in the network.
A broadcast graph with the minimum possible number of edges is called a $minimum\ broadcast\ graph$. However, minimum broadcast graphs are only constructed for some special values of $n$. Therefore, more studies focus on broadcast graph construction (not necessarily minimum) after broadcast graphs were introduced by Kn{\"o}del \cite{Knodel:75} in 1975. Most constructions are classified into the following methods: the compound method \cite{harutyunyan_new_2020,harutyunyan_simple_2019,dinneen_compound_1999,farley_minimal_1979}, the ad-hoc method \cite{harutyunyan_more_1999,bermond_sparse_1992,grigni_tight_1991,farley_minimum_1979}, and the vertex addition or deletion method \cite{harutyunyan_efficient_2008,harutyunyan_more_1999,bermond_sparse_1992}. Certain broadcast graph constructions with specific number of vertices are also studied \cite{bermond_sparse_1992,maheo_minimum_1994,labahn_minimum_1994,zhou_minimum_2001}. 

Although numerous results have focused on the broadcast time problem and broadcast graph constructions, the complexity of verifying whether a graph is a broadcast graph has not yet been discussed.
This task is challenging because a mapping reduce from a known problem $X$ to the broadcast graph problem is needed. This reduction requires transforming yes-instances of $X$ to broadcast graphs and no-instances of $X$ to non-broadcast graphs, which is more difficult than constructing certain broadcast graphs. This paper proves that the broadcast graph problem is NP-complete.
 The broadcast graph problem is formally defined as follows.
\begin{problem}\label{BDP}
\textsc{Broadcast Graph}
\begin{itemize}
\item Input: A graph $G$.
\item Question: Is $G$ a broadcast graph?
\end{itemize}
\end{problem}
Currently, the most related result is given by Slater, Cockayne, and Hedetniemi in \cite{slater_information_1981}. They proved the NP-completeness of \textsc{Broadcast Time} by a reduction from \textsc{3-dimensional Matching}. Other results on NP-completeness of broadcast time problem include work on 3-regular planar graphs \cite{MIDDENDORF1993281}, bounded degree graphs \cite{dinneen1994complexity} and chordal graphs \cite{JANSEN199569}.
\begin{problem}\textsc{Broadcast Time}
\begin{itemize}
\item Input: A three-tuple $(G,V_{0},t)$, where $G=(V, E)$ is a graph, $V_{0}\subseteq V$ is the set of originators, and $t\in \mathbb{N}$ is the given broadcast time bound. 
\item Question: Is $b(G,v) \leq t$ for all $v \in V_{0}$?
\end{itemize}
\end{problem}
We prove \textsc{Broadcast Graph} is NP-complete by a reduction from \textsc{ST-Broadcast Time}, a specialization of \textsc{Broadcast Time}, where $V_0$ is a singleton set of one originator, and $t=\lceil \log |V|\rceil$.
\begin{problem}Single Origin Time-bounded Broadcast Time (\textsc{ST-Broadcast Time})
\begin{itemize}
\item Input: A two-tuple $(G, v)$, where $G=(V, E)$ is a graph, and $v\in V$ is the originator.
\item Question: Is $b(G,v)= \lceil \log |V|\rceil$?
\end{itemize}
\end{problem}
Papadimitriou and Yannakakis \cite{Papadimitriou1982TheCO} proved a similar problem to be NP-complete, which is given a graph $G$, whether one can find a binomial spanning tree of $G$. This problem is equivalent to asking if there exists a vertex $v$ such that $b(G,v)=\lceil \log n\rceil$, while $v$ is specified in \textsc{ST-Broadcast Time}. Therefore, we provide an independent proof for the NP-completeness of \textsc{ST-Broadcast Time} in Appendix \ref{app:st}, by a reduction from \textsc{3-Dimensional Matching}. 
\\
In addition to the \textsc{ST-Broadcast Time} problem, people are also interested in determining the most efficient vertex from which to broadcast in the network. Focusing on these critical nodes leads to the study of the $broadcast\ center$, the set of vertices having the minimum broadcast time in a graph $G$, denoted by $BC_G$. It is demonstrated that both decision problems $Broadcast\ Center\ Deciding\ Problem$ and $Broadcast\ Center\ Size\ Problem$ of determining the broadcast center or its size, are NP-hard \cite{broadcastcenter2021}.

\begin{problem}Broadcast Center Size (\textsc{BC-Size})
\begin{itemize}
\item Input: A two-tuple $(G, x)$, where $G=(V, E)$ is a graph, and $x\in \mathbb{N}$ is the given broadcast center size bound.
\item Question: Is $|BC_G|= x$?
\end{itemize}
\end{problem}
In the present paper we demonstrate that \textsc{BC-Size} is in $\Delta^2_p$ and is $D^P$-hard, which allows us to locate the complexity of \textsc{BC-Size} within an interval, in which the lower bound is improved from NP-hard to $D^P$-hard and the upper bound is $\Delta^2_p$-complete.  
Notably, it is not our concern whether \textsc{BC-Size} is in $D^P$, given that $D^P$ is currently below the lower bound.
\\
$\Delta^2_p$ is a complexity class on the second level of the polynomial hierarchy ($PH$). A problem is in $\Delta^2_p$ if it can be solved by a polynomial-time algorithm with the access to an NP oracle. Thus, $\Delta^2_p$ is also denoted as $P^{NP}$. The difference polynomial time $D^P$ is the complexity class for languages that are in the intersection of a language in NP and a language in $co$-NP. $D^P$ is equivalent to $BH_2$, the second level on the  Boolean hierarchy, and the entire $BH$ is contained in $\Delta^2_p$. Our proof of the $D^P$-hardness is established by a reduction from \textsc{Unique-SAT}, which is known to be $D^P$-complete \cite{blass_unique_1982}. For more details about $PH$ and $BH$, readers are referred to \cite{BH1,BH,Schaefer2008CompletenessIT}.

\begin{problem}Unique Satisfiability (\textsc{Unique-SAT})
\begin{itemize}
\item Input: A CNF boolean formula $\phi(x_0,\cdots,x_{n-1})$ of $n$ variables and $c$ clauses.
\item Question: Is there a unique assignment to satisfy 
$\phi$ ?
\end{itemize}
\end{problem}

This paper is organized as follows. Section 2 reviews some important graph structures used in our proofs. Section 3 proves the NP-completeness of \textsc{Broadcast Graph}. Section 4 studies the complexity of \textsc{BC-Size}.

\section{Preliminaries}
Our proof of NP-completeness for the \textsc{Broadcast Graph} problem involves constructing broadcast graphs (as problem instances) under specific conditions. Therefore, this section reviews some useful graph structures \cite{hl13book} benefitting the proof. \\
To construct broadcast graphs, we use the compounding method introduced by Averbunch, Shabtai and Roditty \cite{Averbuch2013EfficientCO}, and later used by Harutyunyan and Li \cite{harutyunyan_new_2020}, which combines hypercubes or Kn{\"o}del graphs with binomial trees.

\begin{definition}
A binomial tree $BT_{k}=(V,E)$ of degree $k\geq 0$ is defined recursively:
\begin{itemize}
\item (base case) when $k=0$, $BT_{0} = (\{\varepsilon\}, \{\})$;
\item (recursion) for $k\geq1$, $BT_{k}$ rooted at $r$ is consisted of 3 parts, $BT_{k-1}$ rooted at $r$, its replica $BT_{k-1}'$ rooted at $r'$, and an additional edge $(r,r')$. Each vertex is represented by a binary string of $k$ bits in $BT_{k}$, either `1 suffixed by its binary in $BT_{k-1}$', or `0 suffixed by its binary in $BT_{k-1}'$'.
\end{itemize}
\end{definition}
\begin{figure}\centering
	 \def\svgwidth{\textwidth}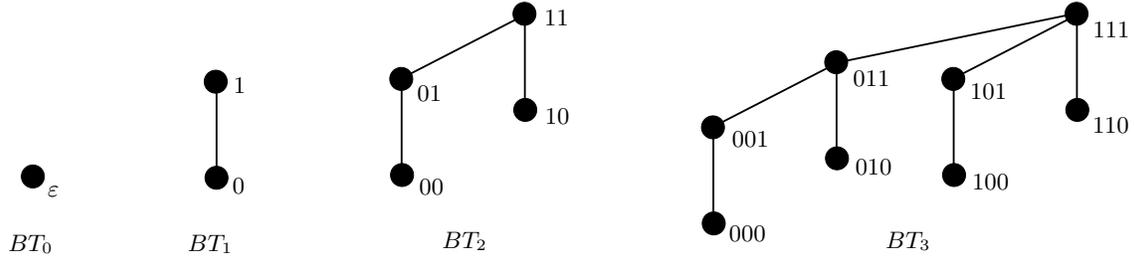
	 \caption{\textbf{binomial trees of degree $0$ to $3$.} Each  vertex is labeled by its binary representation, This notation allows us to sort the leaves (vertices ends with `0') by their distances to the root in descending order, which is useful in Section 4. In $BT_i$, the first $i-1$ bit of the leaves are considered. For example, in $BT_3$, the first two bits of the leaves are $bin(00)=dec(0)$,  $bin(01)=dec(1)$, $bin(10)=dec(2)$, $bin(11)=dec(3)$.}
\end{figure}
We introduce binary notation for vertices in binomial trees. Note $\varepsilon$ is the empty string. A leaf's binary ends with 0. We define $l_i=\lambda_0\lambda_1\cdots\lambda_{d-2}0$ as the $i^{th}$ leaf if $\lambda_0\lambda_1\cdots\lambda_{d-2}$ is the binary representation of $i$.

To broadcast from the root $v_0$ in $BT_k$, the broadcast scheme follows the recursive definition of $BT_k$ that, each root of a $BT_{k-i}$ calls its replica's root at time $i$. Thus, $b(BT_k,v_0)=k$. Since $BT_k$ has $2^k$ vertices, $k=\log 2^k$ is the minimum time for broadcast in $BT_k$, with each vertex busy in this broadcast scheme. 

Since $2^k$ is the maximum number of vertices that can be informed in $k$ time units, any broadcast tree with broadcast time $k$ is a subtree of the binomial tree $BT_k$.

\begin{observation}\label{leaf}
A multicast from the root $v_0$ in $BT_k$ to its first any $x$ leaves $\{l_0,\cdots,l_{x-1}\}$ requires time $k$.
\end{observation}
This is because $\{l_0,\cdots,l_{x-1}\}$ contains $l_0$, and the distance $d(v_0,l_0)=k$, the multicast time cannot be smaller than $k$.

\begin{definition}
A Kn{\"o}del graph $KG_n=(V,E)$ is defined on $n = 2k$ vertices for $k\geq 0$, where
\begin{itemize}
\item $V=\{v_0, v_1, \cdots, v_{n-1}\}$;
\item $E=\{(v_x,v_y)|x+y\equiv 2^d-1 \mod n, 1 \leq d \leq \lfloor \log n\rfloor\}$. Each edge $(v_x,v_y)$ corresponds to dimension $d$.
\end{itemize}
\end{definition}
Kn{\"o}del graphs are regular and edge-transitive. Hence, a dimensional broadcast scheme is applicable, except that in the final time unit $\lceil \log n\rceil$, all vertices call their $1^{st}$-dimensional neighbor.

\begin{observation}\label{o1} In a Kn{\"o}del graph, both the originator and its $1^{st}$-dimensional neighbor are idle in the last time unit.\end{observation}
This can be  easily demonstrated through the dimensional broadcast scheme, a property that makes Kn{\"o}del graphs valuable in broadcast graph constructions \cite{harutyunyan_more_1999,harutyunyan_new_2020,harutyunyan2023broadcast}.

\section{Broadcast Graph Problem is NP-complete}
In this section, we prove the NP-completeness of the \textsc{Broadcast graph} problem by performing a polynomial time many-one reduction (\textsc{ST-Broadcast Time} $\leq_p$ \textsc{Broadcast Graph}) from the known NP-complete problem \textsc{ST-Broadcast Time}.  We demonstrate that \textsc{Broadcast Graph} is both in NP and NP-hard.

\begin{lemma}\label{lm1}
\textsc{Broadcast Graph} is in NP.
\end{lemma}

\begin{proof}
We prove this lemma by presenting a polynomial-time verifier.
Let $G=(V,E)$ be an instance of \textsc{Broadcast Graph} with $n$ vertices, and $V=\{ v_1,v_2,...,v_n\}$. The certificate is a set of spanning trees $S=\{T_1,T_2, ..., T_n\}$ of $G$, such that for $1\leq i\leq n$, each $T_i$ is rooted at $v_i$, and represents a valid broadcast scheme with originator $v_i$. To verify whether $T_i$ is a subtree of a degree $n$ binomial tree $BT_{\lceil \log n \rceil }$, the time complexity is $O(n)\times(O(n^{5/2})+O(\log n))$, which is $O(n^{7/2})$ \cite{matula_subtree_1978}.
\end{proof}
For the reduction, we construct an instance of \textsc{Broadcast Graph} from an \textsc{ST-Broadcast Time} instance and prove that this construction is a polynomial-time reduction. 
\begin{definition}\label{def3}
	Given an instance of STBT, $(G_s,v_s)$ on $n_s$ vertices, a graph $G_u$ is constructed by the following algorithm.
	\begin{enumerate}
		\item Create a Kn{\"o}del graph $KG_6$ on 6 vertices labeled $v_0, v_1, \cdots ,v_5$.
		\item Create 5 copies of binomial trees $T_1, \cdots , T_5$ of degree $\lceil \log {n_s}\rceil+1$ rooted at $r_1, \cdots, r_5$.
		\item Create 1 binomial tree $T_6$ of degree $\lceil \log {n_s}\rceil$ rooted at $r_6$.
		\item Merge each pair of vertices $(v_i, r_i)$ for $1\leq i\leq 5$ and merge $(v_s, v_0)$.
		\item Add edges $(u,r_1)$ and $(u,r_5)$ for every vertex $u$ in $G_s$, $T_1$, $T_5$, and $T_6$ (if not already adjacent). Also add an edge $(v_s,r_6)$.
		\item Add edges $(u,r_2)$ and $(u,r_4)$ for every vertex $u$ in $T_2$, $T_3$, and $T_4$ (if not already adjacent).
	\end{enumerate}
\end{definition}
\noindent From the construction, we calculate the number of vertices.
\begin{align*}
	|V_u|	& = 5(2^{\lceil \log n_s\rceil+1}) && \text{vertices in $BT_1, \cdots, BT_5 $ }\\
		& \quad + 2^{\lceil \log n_s\rceil} && \text{vertices in $BT_6$ }\\
			  & \quad + n_s && \text{the vertices from $G_s$}\\
			& = 11(2^{\lceil \log n_s\rceil}) + n_s
\end{align*}

\noindent Let $\lceil \log{n_s}\rceil = t$ for some $t\in\mathbb{N}$. 
\begin{align*}
2^{t-1}+1 \leq & |V_s|\leq 2^{t}\\
2^{t}\times 11 + 2^{t-1}+1 \leq & |V_u|\leq 2^{t}\times 11 +2^{t}\\
2^{t+3} = 2^t \times 8 < & |V_u|< 2^t \times 16= 2^{t+4}
\end{align*}
Thus, if $G_u$ is a broadcast graph (yes-instance of \textsc{Broadcast Graph}), then the broadcast time $b(G_u)=\lceil \log |V_u| \rceil = t+4 = \lceil \log n_s\rceil+4$. 
The number of edges in $G_u$ is calculated as follows.
\begin{align*}
	|E_u|	& = |E_s| && \text{the edges from $G_s$}\\
		& \quad + 5(2^{\lceil \log n_s\rceil+1}) && \text{edges in $BT_1, \cdots, BT_5$}\\
		& \quad + 2^{\lceil \log n_s\rceil} && \text{edges in $BT_6$}\\
		& \quad + 2(n_s)+4(2^{\lceil \log n_s\rceil+1})+2(2^{\lceil \log n_s\rceil})+1 &&\text{step 5}\\
		& \quad +6(2^{\lceil \log n_s\rceil+1}) && \text{step 6}\\
		& = 33(2^{\lceil \log n_s\rceil})+2n_s+|E_s|+1
\end{align*}
$G_u$ is polynomial-time constructible on $G_s$. Next, we prove from both directions(yes and no instances) that the construction is a reduction from $b(G_s,v_s)$ to $G_u$.

\begin{figure}
\centering
	 \def\svgwidth{0.9\textwidth}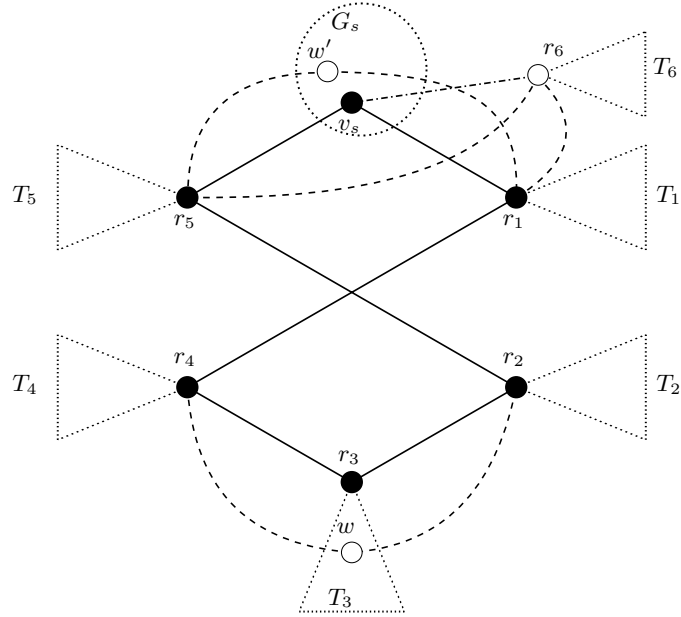
	 \caption{\textbf{Example construction of $G_u$ on $G_s$.} The dotted triangles represent the binomial trees $T_1, \cdots, T_5$ of degree $\lceil\log{n_s}\rceil+1$, and $T_6$ of degree $\lceil\log{n_s}\rceil$. A Kn{\"o}del graph $KG_6$ is formed by $r_1, \cdots, r_5$, and $v_s$. The dashed lines represent edges added in step 5 and 6 in Definition 6. The dot-dashed line is $(v_s,r_6)$.
	 }\label{fig:BG}
\end{figure}
\begin{lemma}\label{lm2}
$(G_s, v_s)$ is a yes-instance of $STBT$ if and only if $G_u$ is a broadcast graph.
\end{lemma}
\begin{proof}
$(\Rightarrow)$ Assume $(G_s, v_s)$ is a yes-instance of \textsc{ST-Broadcast Time}, that is, $b(G_s, v_s) = \lceil \log n_s \rceil = t$. We apply a broadcast scheme from \cite{harutyunyan_more_1999} for each originator position in $G_u$. There are 3 cases for originators.
\begin{enumerate}
\item If the originator is an arbitrary vertex in $KG_6$, it informs all vertices in $KG_6$ in the first 3 time units using the dimensional broadcast scheme of Kn{\"o}del graphs. From time 4, $r_1,\cdots ,r_5$ broadcast in $T_1,\cdots, T_5$ separately, completing in $3+t+1=t+4$. Simultaneously, $v_s$ informs $r_6$ by time $4$. From time $5$ to $t+4$, $v_s$ completes the broadcast in $G_s$, and $r_6$ in $T_6$.
\item If the originator (not in $KG_6$) is an arbitrary vertex in $T_1, T_5, T_6, G_s$, for example, $w'$ in Fig. \ref{fig:BG}, then $w'$ mimics $v_s$ to broadcast in $KG_6$ in the first 3 time units. This is similar to the dimensional broadcast, except that $v_s$ is informed by $r_5$ or $r_1$ at time 3. By Observation \ref{o1}, either $r_5$ or $r_1$ is idle in the last time unit. The rest follows case 1.
\item If the originator (not in $KG_6$) is an arbitrary vertex in $T_2, T_3, T_4$, for example $w$ in Fig. \ref{fig:BG}, the broadcast follows case 2 except $w$ mimics $r_3$ in $KG_6$, informed by $r_2$ or $r_4$ in time unit 3.
\end{enumerate}
\noindent 
$(\Leftarrow)$ If $(G_s, v_s)$ is a no-instance of \textsc{ST-Broadcast Time}, which means $b(G_s, v_s) > \lceil \log n_s \rceil = t$, then exists $w\in G_u$ such that the broadcast from $w$ in $G_u$ cannot finish in $t+4$ time units. Assume $w$ is an arbitrary non-root vertex in $T_3$, the proof enumerates all possible broadcast schemes from $w$ in $G_u$.\\

First, note $\{(r_1,r_4),(r_2,r_5)\}$ is a cut in $G_u$, with $d(w,r_1)=d(w,r_5)=2$. 
Thus, originating from $w$, $r_1$ and $r_5$ cannot both be informed by time 2. One of $r_1$ and $r_5$ can be informed by time 3, the earliest possible time. With $r_1$ and $r_5$ being symmetric in $G_u$(there is an automorphism $f:G_u\to G_u$ such that $f(r_1)=r_5$), we assume $r_1$ and $r_5$ are informed in time 2 and 3. At time 3, the broadcast from $w$ splits into two parallel independent multi-originator broadcast sub-schemes: one in $T_1\cup T_5\cup T_6\cup G_s$ from $\{r_1,r_5,x_1\}$, where $x_1$ is a neighbor of $r_1$; another in $T_2\cup T_3\cup T_4$ from $\{w,r_2,r_4,x_2,x_3\}$, with $x_2$ as a neighbor of $w$ and $x_3$ of $r_4$.
Since no vertex in one sub-scheme can inform the other (if not, the broadcast time trivially exceed $t+4$), if the first sub-scheme cannot finish in $t+4-3=t+1$, then $G_u$ is not a broadcast graph.

Next, for the first sub-scheme, consider $T_1$ and $T_5$. Broadcasting from $r_5$ in $T_5$ needs $t+1$ time units - the exact remaining time for $r_5$. Similarly, $r_1$ needs $t+1$ units in $T_1$. Vertices in $T_1$ and $T_5$ are kept busy in time $4$ to $t+4$. Thus, the question is which neighbor $x_1$ is informed by $r_1$ at time 3?  Besides in binomial trees $T_1$ or $T_5$, there are four choices:
\begin{enumerate}
\item If $r_1$ informs $v_s$ at time 3, then
	\begin{itemize}
		\item if $v_s$ informs $r_6$ in time 4, broadcasting in $G_s$ cannot finish in $t+4$ due to $b(G_s,v_s)>t$;
		\item if $v_s$ does not inform $r_6$ in time 4, broadcasting in $T_6$ cannot finish in $t+4$, as $T_6$ is a degree $t$ binomial tree.
	\end{itemize}
\item If $r_1$ informs $r_6$ at time 3, it is similar to case 1.
\item If the neighbor $x_1$ is in $T_6$ (not $r_6$), then $d(v_s,x_1)\geq 2$, $v_s$ cannot be informed by time 4, and broadcasting in $G_s$ cannot finish in $t+4$.
\item If the neighbor $x_1$ is in $G_s$ (not $v_s$), then $d(r_6,x_1)\geq 2$, $r_6$ cannot be informed by time 4, and broadcasting in $T_6$ cannot finish in $t+4$. 
\end{enumerate}
Thus, if $(G_s, v_s)$ is a no-instance of \textsc{ST-Broadcast Time}, $G_u$ is not a broadcast graph.
\end{proof}
\noindent Summarizing lemmas \ref{lm1} and  \ref{lm2}, and definition \ref{def3} is a polynomial-time construction,  
\begin{theorem}\label{bupnpc}
\textsc{Broadcast Graph} is NP-complete.
\end{theorem}

\section{The Complexity of \textsc{BC-Size} Problem}
In this section, we examine the complexity of \textsc{BC-Size}. We demonstrate that \textsc{BC-Size} is in $\Delta^p_2$ and $D^P$-hard by reducing from \textsc{Unique-SAT}, which is $D^P$-complete \cite{PAPADIMITRIOU1984244}.
\begin{lemma}\label{td}
\textsc{BC-Size} is in $\Delta^p_2$.
\end{lemma}
\begin{algorithm}[H]
\caption{Polynomial algorithm for deciding \textsc{BC-Size} with access to the oracle \textsc{Broadcast Time}.}
\LinesNumbered
\KwIn{ $G=(V,E),\ c\in \mathbb{N}$}
\KwOut{ $1$ if $|BC_G|=c$; $0$ otherwise. }
\For{$i=\log \lceil |V| \rceil$ \KwTo $n$}{
	size = 0\\
	\For{$v\in V$}{
		\If{$(G,\{v\},i)\in \textsc{Broadcast Time}$}{size++}
	}
	\If{$size != 0$}{
		\If{$size == c$}{\Return 1}
		\Else{\Return 0}
	}
}
\Return 0
\end{algorithm}
\begin{proof}
The algorithm uses an NP-complete oracle for \textsc{Broadcast Time} on line 4 and runs in $O(n^2)$. 
\end{proof}

For the $D^P$-hardness, we construct a \textsc{BC-Size} instance based on a graph representation of a \textsc{Unique-SAT} instance.
\begin{definition}\label{def4}
Consider a CNF formula $\phi(x_0,\cdots,x_{n-1})$, an instance of \textsc{Unique-SAT} with $c$ clauses, $\delta_1,\delta_2,\cdots,\delta_c$. A graph $G$ is constructed by following algorithm.\begin{enumerate}
\item Create $2n$ copies of binomial trees $T^{x_0},T^{\overline{x_0}},\cdots,T^{x_{n-1}},T^{\overline{x_{n-1}}}$ of degree $d_1=\lceil \log c\rceil +1$, rooted at literals $x_0,\overline{x_0},\cdots,x_{n-1},\overline{x_{n-1}}$ respectively. 
\item Create $2n$ copies of binomial trees $H^{x_0},H^{\overline{x_0}},\cdots,H^{x_{n-1}},H^{\overline{x_{n-1}}}$ of degree $d_2=\lceil \log n\rceil +1$, rooted at $r^{x_0},r^{\overline{x_0}},\cdots,r^{x_{n-1}},r^{\overline{x_{n-1}}}$ respectively.
\item Create a binomial tree $T^r$ of degree $d=d_1+d_2+1=\lceil \log c\rceil+\lceil \log n\rceil+3$, rooted at $r$. 
\item Construct a star with $\{r^{x_0},r^{\overline{x_0}},\cdots,r^{x_{n-1}},r^{\overline{x_{n-1}}},r,s\}$, where $s$ is the center.
\item Assume that the literal $\hat x$ of variable $x$ is contained in $j$ clauses $\delta_1',\delta_2',\cdots,\delta_j'$. Construct edges $(\lambda_i^{\hat x},\delta_j')$ for $0\leq i<j$ where $\lambda_i^{\hat x}$ is the $i^{th}$ leaf of $T^{\hat x}$. And do the same for every literal.
\item For each $H^{\hat x_i}$, connect each $k^{th}$ leaf to both $x_k$ and $\overline{x_k}$, except the $i^{th}$ leaf to only $x_i$. Formally, $\{(l_k^{\hat x_i},x_{k})|0\leq k \leq n-1\}\cup\{(l_k^{\hat x_i},\overline{x_{k}})|0\leq k \leq n-1, x_k\neq \hat x_i\}$. 
\item  Connect each of the first $c$ leaves of $T^r$ to one distinct clause. Formally, $\{(l_{i-1},\delta_{i})|1\leq i\leq c\}$.
\item Attach $2n-\lceil\log c\rceil -2$ vertices to each $\delta_i$ to form $c$ stars of size $2n-\lceil\log c\rceil -1$ (equivalent to $2n-d_1$) with center $\delta_i$.
\item Add a path $P=(p_0,p_2,\cdots,p_{2n+\lceil\log n\rceil + 1})$. Connect the star center $s$ to the path $P$ with an edge $(s,p_0)$.
\end{enumerate}
\end{definition}
\begin{figure}
\centering
\def\svgwidth{\textwidth}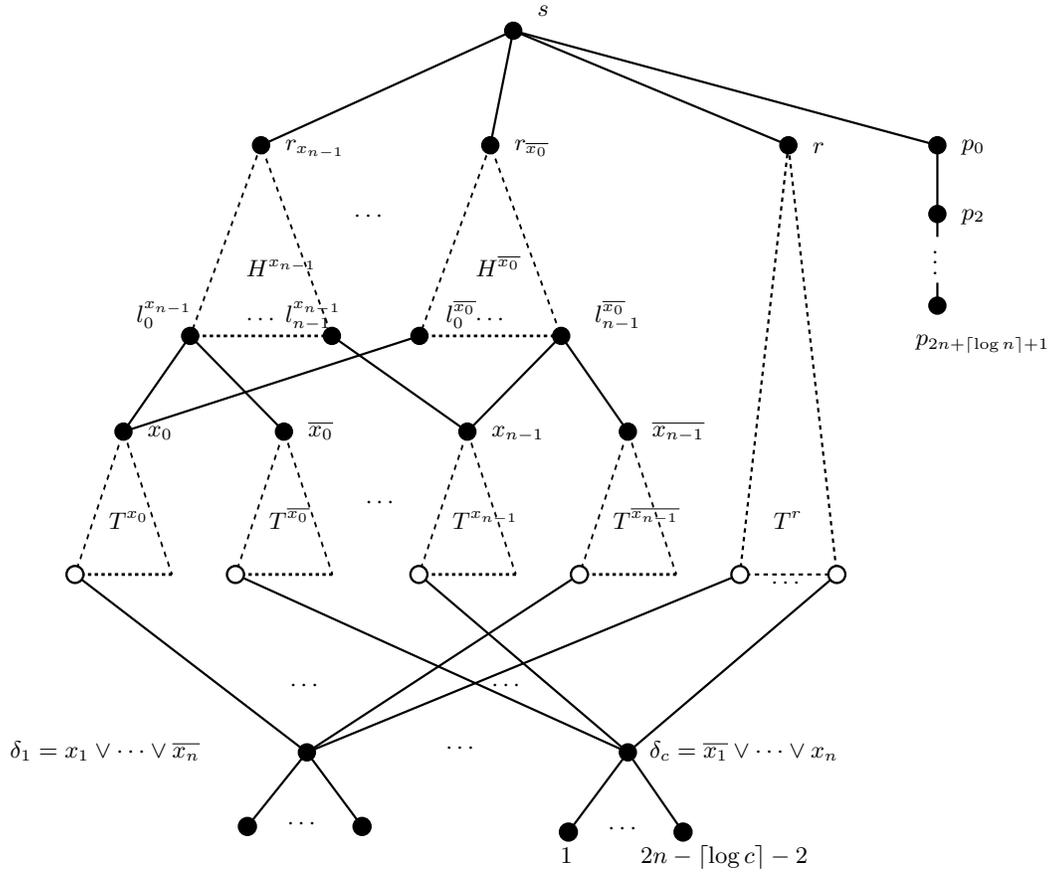
   \caption{\textbf{An example of the construction of $G$.} The binomial tree $H^{\hat{x_i}}$ represents the assignment $\hat{x_i}=1$. Edges between $H$'s and $T$'s show the two possible assignments for each variable. Each binomial tree $T^{\hat{x_i}}$ presents a literal $\hat{x_i}$. Each $\delta_i$ is a clause. The binomial tree $T^r$ ensures $s$ and $r$ are in the broadcast center. The path $p_0,\cdots$ forces all vertices other than $s,r,r_{x_{0}},\cdots$ are out of the broadcast center.}\label{fig:main}
\end{figure}
\noindent Figure \ref{fig:main} shows an example of the construction. The construction is polynomial-time. \\
The variable gadget $\hat{x_i}$ is a single vertex adjacent to two literals $x_i$ and $\overline{x_i}$. The clause gadget $\delta_k$ is a vertex adjacent to $x_i$(or $\overline{x_i}$) if $x_i$(or $\overline{x_i}$) is a literal in $\delta_k$. Assignments on $\hat{x_i}$ are represented by choosing either $x_i$ or $\overline{x_i}$ for broadcast. For example, if $x_i=0$, then $\hat{x_i}$ calls $\overline{x_i}$, which calls its clauses (see Figure \ref{fig:c} for an example), and vice versa. A literal can be in multiple clauses, so each $T^{x_i}$ (or $T^{\overline{x_i}}$) of degree $d_1$ connects to each literal $x_i$ (or $\overline{x_i}$) adjacent to theirs leaves. Broadcasting from the set of all variables (as multi-originator) to all clauses completes in $d_1+2$ time units. \\
The construction also ensures a ``unique'' assignment on $\phi$, achieved by selecting each essential literal among the $2n$ literals. In $G$, this is represented by a binomial tree $H^{\hat{x_i}}$ rooted at $r_{\hat{x_i}}$. Selecting $x_i$ forces it to be $True$, so $H^{\hat{x_i}}$'s leaves connect only to $x_i$, not $\overline{x_i}$. Broadcasting from $r_{\hat{x_i}}$ completes in $d_1+d_2+3$ time units if and only if $\phi$ is satisfiable with $\hat{x_i}=True$ (see Figure \ref{fig:c} for example). The rest steps of the construction create the following. \begin{itemize}
\item A star centered at $s$ to originate a broadcast.
\item A binomial tree $T^r$ ensures $s$ and $r$ are in the broadcast center even if $\phi$ is not satisfiable.
\item A path $P$ simplifies the proof.
\item Stars centered at clauses are to fill the broadcast time.
\end{itemize}
\begin{figure}
\centering
\def\svgwidth{\textwidth}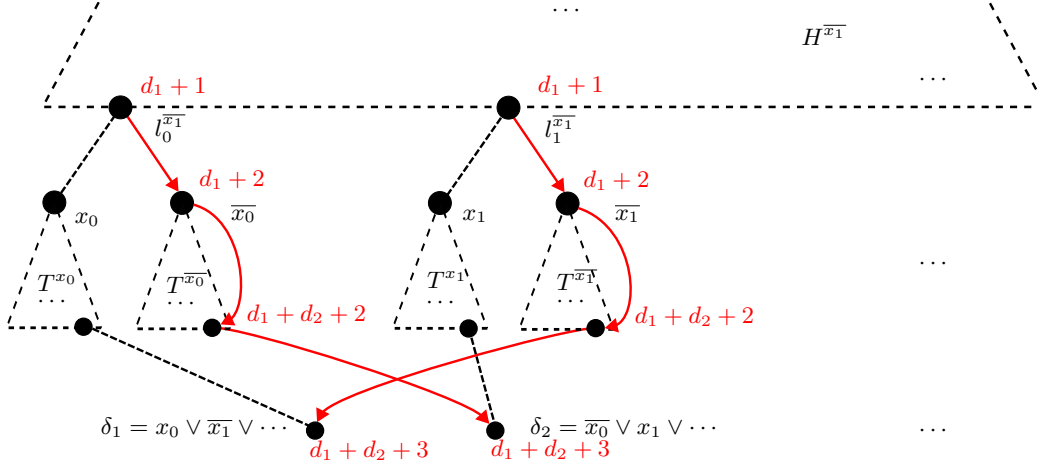
   \caption{\textbf{An example for variable, literal, and clause gadgets.} Assume $\delta_1=x_0\vee \overline{x_1}\vee\cdots$ and $\delta_2=\overline{x_0}\vee x_1\vee\cdots$ are two clauses in $\phi$. Vertex $\delta_1$ is adjacent to one leaf in $T^{x_0}$ and another leaf in $T^{\overline{x_1}}$ (same for $\delta_2$). If all leaves of $H^{\overline{x_1}}$ are informed at time unit $d_1+1$, then directed edges and time units (in red) show the calls from $l_0^{\overline{x_1}}$ and $l_1^{\overline{x_1}}$ to $\delta_1$ and $\delta_2$ respectively at $d_1+d_2+3$, mimicking the assignment $x_0=False$ and $x_1=False$ to satisfy both $\delta_1$ and $\delta_2$.}\label{fig:c}
\end{figure}

\begin{lemma}\label{lm4}
There exists a unique assignment to $\phi(x_0,\cdots,x_{n-1})$ if and only if the graph $G$ constructed above has a broadcast center of size $n+2$. 
\end{lemma}
\begin{proof}
Let $t = 2n+\lceil \log n\rceil +3$. The proof consists of four steps.\begin{enumerate}
\item We first show that \begin{align*}b(G,s)=b(G,r)=t, \end{align*}using the following broadcast scheme.
In time unit 1, $s$ calls $r$ (or vice versa). From time unit $2$ to $d_1+d_2+1$, $r$ broadcasts within $T^r$. All clauses are informed by time $d_1+d_2+2$, and all vertices in stars are informed by time $t$ ($(\lceil \log c\rceil+1)+(\lceil \log n\rceil+1)+3+(2n-\lceil \log c\rceil-2)$).\\
Simultaneously, $s$ calls $p_0$ at time 2, and the broadcast along the path completes by time $t$. Starting at time 3, $s$ calls all other $r^x$'s. The last $r^x$ is informed by time $n+2$ and its binomial tree $H^x$ completes by time $t$. 

\item Next, we prove that,\begin{align*} \text{If}\ \phi\ \text{is\ satisfiable\ when}\ x_i=1\text{,\ then}\ b(G,r_{x_i})=t.\end{align*} This case is similar to the first, with the difference occurring at time $d_1$, when the broadcast in $H^{x_i}$ finishes. If there is a unique satisfying assignment, the leaves of $H^{x_i}$ select $n-1$ literals of the assignment (including $x_i$). The $n$ $T$-trees, rooted at these literals, are informed by time $d_1+1$. Consequently, all clauses are informed by time $d_1+d_2+2$. The rest steps follow case 1.

\item We also need to show that, \begin{align*}\text{If}\ \phi\text{\ is\ not\ satisfiable\ when}\ x_i=1,\ \text{then}\ b(G,r_{x_i})>t.\end{align*} Consider the broadcast from $r_{x_i}$, it must inform $s$ in the first time unit because the distance $d(r_{x_i},p_{2n+\lceil \log n \rceil +1})=t$. If the satisfying assignment is unique, since $r_{\overline{x_i}}$ is unreachable in $d_1+1$, there exists some clause $\delta_k$ cannot be informed by time $d_1+d_2+2$. Thus, completing the star at $\delta_k$ requires at least $t+1$ time.

\item Finally, \begin{align*}\text{For\ all\ vertices}\ v\ \text{not\ in\ previous\ cases,}\ b(G,v)>t.\end{align*}
If $v$ is not on the path, then $b(G,v)\geq d(v,p_0)+b(P,p_0)= d(v,r)+d(r,p_0)+b(P,p_0)\geq 1+2+b(P,p_0)=2n+\lceil\log n\rceil+4=t+1$.\\
If $v$ is on the path, then $b(G,v)\geq d(v,\delta_1)+b(Star,\delta_1) = d(v,s)+d(s,\delta_1)+b(Star,\delta_1) \geq 1+1+d_1+d_2+3+1+(2n-d_1)=2n+\lceil\log n\rceil+5=t+2$.\\
\end{enumerate}
According to the above, $b(G)=t$. If $\phi$ is uniquely satisfiable when $x_i=1$, then $b(G, r_{x_i})=b(G)=b(G,s)=b(G,r)=t$, so $|BC_G|=n+3$. If there are multiple satisfying assignments, some vertices in case $3$ are also in the broadcast center, making $|BC_G|>n+3$. If there is no satisfying assignment, $|BC_G|=|\{s,r\}|=2$. In conclusion, $s$ and $r$ are always in the broadcast center, and $r_{x_i}$ is included if and only if $\phi$ is satisfiable when $x_i=1$. 
\end{proof}
\noindent Summarizing lemmas \ref{td} and \ref{lm4}, and definition \ref{def4} is a  polynomial-time construction,
\begin{theorem}
\textsc{BC-Size} is $\Delta^p_2$ and $D^P$-hard.
\end{theorem}

\section{Conclusion}
In this paper, we demonstrated that the \textsc{Broadcast Graph} problem is NP-complete by reducing it from \textsc{ST-Broadcast Time}. The NP-completeness proof of \textsc{ST-Broadcast Time} is also provided in the appendix. 
Moreover, we further investigated the complexity of the NP-hard problem \textsc{BC-Size}. We show that \textsc{BC-Size} is $D^P$-hard by reducing it from \textsc{Unique-SAT}. Present work provides a finer complexity interval for the problem, which is between $D^P$-hard and $\Delta_p^2$-complete. In future work, we aim to precisely locate the complexity of \textsc{BC-Size} within the boolean hierarchy($BH$) or prove its $\Delta_p^2$-completeness.

\appendix
\setcounter{figure}{0}
\setcounter{definition}{0}
\setcounter{lemma}{0}
\setcounter{theorem}{0}
\setcounter{problem}{0}
\renewcommand{\thefigure}{A\arabic{figure}}
\renewcommand{\thedefinition}{A\arabic{definition}}\renewcommand{\thelemma}{A\arabic{lemma}}
\renewcommand{\thetheorem}{A\arabic{theorem}}
\renewcommand{\theproblem}{A\arabic{problem}}

\section{Single Origin Time-bounded Broadcast Time Problem is NP-complete}\label{app:st}

\begin{lemma}\label{a1}
\textsc{ST-Broadcast Time} is in NP.
\end{lemma}

\begin{proof}
We prove the lemma by presenting a polynomial-time verifier. 
Let $(G,v)$ be an instance of \textsc{ST-Broadcast Time}, where $G=(V,E), |V|=n$, and $v$ is the single originator. The certificate for the instance is a spanning tree $T$ of $G$, rooted at $v$, which intends to represent a valid broadcast scheme. Verifying the certificate that whether $T$ is a subtree of $BT_{\lceil \log n \rceil}$, the time complexity would be $1\times (O(n^{5/2})+O(\log n))$, which is $O(n^{5/2})$ \cite{matula_subtree_1978}.
\end{proof}

To prove the NP-hardness, we present a reduction from the problem \textsc{3-Dimensional Matching}. 
\begin{problem}\textsc{3-Dimensional Matching}
\begin{itemize}
\item Input: A five-tuple $(X, Y, Z, W, k)$, where \begin{itemize}
\item $X, Y, Z$ are three disjoint sets of cardinality $k$, and 
\item $W\subseteq X\times Y\times Z$. 
\end{itemize} 
\item Question: Is there a set $M \subseteq W$ of size $k$, such that for two arbitrary elements $(x_1,y_1,z_1),(x_2,y_2,z_2)\in M$, $x_1\neq x_2$, $y_1\neq y_2$, and $ z_1\neq z_2$?
\end{itemize}
\end{problem}

Similar to the original proof \textsc{Broadcast Time} is NP-complete \cite{slater_information_1981}, we present each \textsc{3-Dimensional Matching} instance as a graph.
\begin{definition}
Assume $(X,Y,Z,W,k)$ is an instance of $3DM$, $G_0=(V_0,E_0)$ is the graph representation s.t.\begin{itemize} \item $V_0 = X\cup Y \cup Z\cup W$, and\item $E_0 = \{((x,y,z),s)|(x,y,z)\in W,\ s=x\ or\ s=y\ or\ s=z\}$. \end{itemize}
\end{definition}

We also need to define a special compounding, denoted as a binary operator $\times_t$.
\begin{definition}
Let $G=(V,E)$ be a graph, $V'=\{v_1,v_2,...,v_n\}$ be a subset of $V$, $T$ be a tree rooted at $r_0$ and $T_1,...,T_n$ be $n$ replicas of $T$ rooted at $r_1,...,r_n$ respectively. The compound $V'\times_t T$ is a graph $G\cup T_1\cup T_2\cup ... \cup T_n$ such that merging $r_i$ with $v_i$.
\end{definition}

We further define a new family of trees $BT_k^{-i}$ based on binomial trees for $i\in \mathbb{N}$, $ i\leq k$. Note that a binomial tree $BT_k$ can be regarded as a root with $k$ branches of $BT_{k-1},BT_{k-2},...,BT_0$. Assume that the branches are ranked by the order of binomial trees, $BT_k^{-i}$ is obtained from $BT_k$ by removing the “smallest" $i$ branches.
 Formally, $BT_k^{-i}$ can also be defined recursively. 
\begin{definition}
Assume the root of a tree $T_i^j$ is denoted as $r_i^j$,
 $BT_k^{-i}$ is
\begin{itemize}  
\item $(\{r_k\},\{ \})$, if $i=k$;
\item $(V(BT_{k-1})\cup V(BT_{k-1}^{-i}), E(BT_{k-1})\cup E(BT_{k-1}^{-i}) \cup \{ r_{k-1}, r_{k-1}^{-i}\}),$ and is rooted at $r_{k-1}^{-i}$, if $i<k$.
\end{itemize}
\end{definition}

The construction for the reduction on graph representation $G_0=(V_0,E_0)$ of \textsc{3-Dimensional Matching} instance $(W,X,Y,Z,k)$ has following intuition. First, we construct a binomial tree $BT_{\lceil \log w \rceil+1}$, where $w=|W|$. Next, the leaf set $L$ of the binomial tree is partitioned into $(L_1,L_2,L_3)$ such that $|L_1|=k, |L_2|=w-k$, and $L_3$ contains the rest $2^{\lceil \log w\rceil}-w$ leaves. Then, $L_1$ is compounded with $BT_3$ and $L_2$ is compounded with $BT_4^{-1}$. Consider $G_0$ as a subgraph in the construction, the next step adds edges to let $(L_1\cup L_2,W)$ form a fully connected bipartite graph. And the last step adds vertices and edges to guarantee that every $x_i\in X$ is adjacent to two pendant vertices $x_i'$ and $x_i''$; and every $y_i\in Y$ is adjacent to one pendant vertex $y_i'$. To be noted, this construction includes the graph in \cite{slater_information_1981} as a substructure. Formally, \begin{definition}\label{defa}
Let $G_f=(V_f,E_f)$ be a graph such that
\begin{align*} 
V_f\ =&\ V_0\cup V(BT_{\lceil \log w \rceil+1}) \cup V(L_1\times_t BT_3) \cup V(L_2\times_t BT_4^{-1}) \cup X' \cup X'' \cup Y'\\
E_f\ =&\ E_0\cup E(BT_{\lceil \log w \rceil+1})\cup E(L_1\times_t BT_3)\cup E(L_2\times_t BT_4^{-1}) \cup \\&\ \{(u,v)|u\in L_1\cup L_2, v\in W\}\cup \{(x_i,x_i')|x_i\in X,x_i'\in X'\cup X'', 1\leq i \leq k\} \cup \\&\ \{(y_i,y_i')|y_i\in Y, y_i'\in Y'\cup Y'',1\leq i \leq k\},
\end{align*}
where $X'$ and $X''$ are replicas of $X$, and $Y'$ is a replica of $Y$; and all other notations are as what we have defined above. $(G_f,v_0)$ is an $STBT$ instance.
\end{definition}
To be noted, $G_f$ is undefined when $w < k$ because $(W,X,Y,Z,k)$ is a yes-instance of $3DM$ only when $w\geq k$. We can trivially map those instances to no-instances of $STBT$.
figure \ref{fexample} gives an example of the construction of $G_f$.
\begin{figure}
\begin{center}
    \def\svgwidth{0.8\textwidth}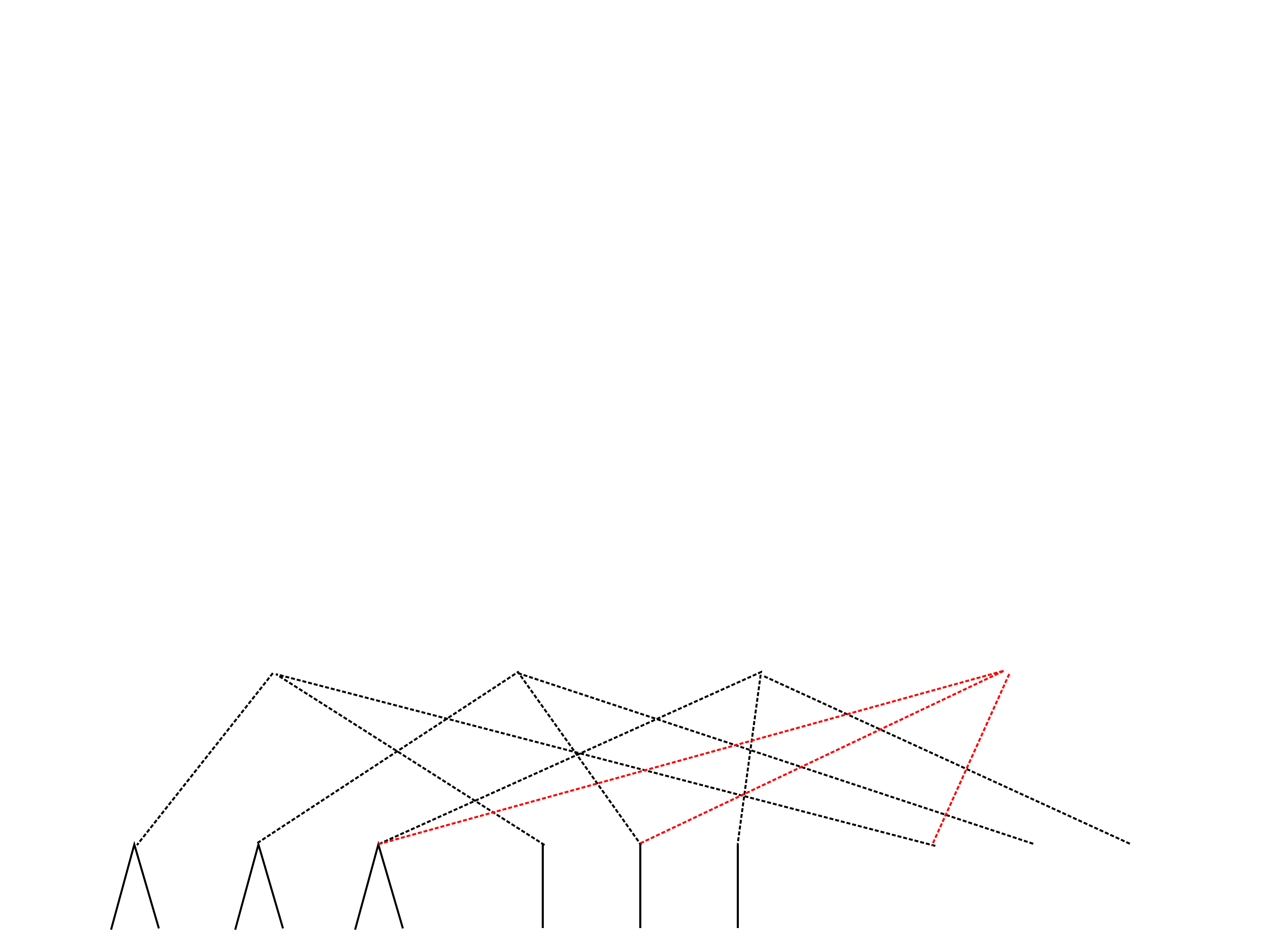
    \caption{\textbf{An example of the construction of $G_f$ respected to the $3DM$ instance $(X=\{x_1,x_2,x_3\}$, $Y=\{y_1,y_2,y_3\}$, $Z=\{z_1,z_2,z_3\}$, $W=\{(x_1,y_1,z_1)$,$(x_2,y_2,z_2)$,$(x_3,y_3,z_3)$,$(x_3,y_2,z_1)\},k=3)$.} The grey vertex is the originator $v_0$.}
\label{fexample}
\end{center}
\end{figure}
\begin{lemma}\label{a2}
The input $(X,Y,Z,W,k)$ is a yes-instance of $3DM$ if and only if $G_f$ is a yes-instance of $STBT$.
\end{lemma}

\begin{proof}
$(\Rightarrow)$ Assume $(X,Y,Z,W,k)$ is a yes-instance of $3DM$, that there exists a 3-dimensional matching $M\subseteq W$ of size $k$. From \cite{slater_information_1981}, we know that $b(G_0,M)=3$. Let $G_f$ be a graph defined above respected to $(X,Y,Z,W,k)$. Then, the broadcast originated from $v_0$ in $G_f$ is as follows.
\begin{enumerate}
\item From time unit 1 to $\lceil \log w \rceil+1$, $v_0$ broadcasts in $BT_{\lceil \log w \rceil+1}$ following the broadcast scheme of the binomial tree. 
\item From time unit $\lceil \log w \rceil+2$ to $\lceil \log w \rceil+5$, the broadcast is split into two parallel processes. 
\begin{enumerate}
\item Each leaf $l$ in $L_1$, which is a root of $BT_3$, informs its neighbor in $W$ at time unit $\lceil \log w \rceil+2$. At this time $k$ vertices in $W$ are informed in total, and the rest of $w-k$ vertices are informed by case $(b)$. Then each $l$  broadcasts in $BT_3$, and each of the $k$ informed vertices in $W$ broadcasts in $G_0$ separately, in time units $\lceil \log w \rceil+3$ to $\lceil \log w \rceil+5$.
\item Each leaf $l'$ in $L_2$, which is a root of $BT_4^{-1}$, is forced to broadcast in $BT_4^{-1}$ in $\lceil \log w \rceil+2$ to $\lceil \log w \rceil+4$. In $\lceil \log w \rceil+5$, vertices in $L_2$ are released to call the $w-k$ uninformed vertices $W$.
\end{enumerate}
\end{enumerate}
According to above broadcast scheme, the broadcast time from the originator $v_0$, $b(G_f,v_0) =\lceil \log w\rceil+5 =\ \lceil \log |V_f| \rceil$,  that $G_f$ is a yes-instance of $STBT$.

$(\Leftarrow)$Assume $(X,Y,Z,W,k)$ is a no-instance of $3DM$, which means that no 3-dimensional matching $M\subseteq W$ exists. In the graph representation $G_0$, the broadcast time from any subset $S$ satisfying $S\subseteq W$ and $|S|=k$, to $S\cup X\cup Y\cup Z$ is greater than 3. Hence, the broadcast from $v_0$ in $G_f$ must takes the first $\lceil \log w \rceil+1$ time units to accomplish the broadcast in $BT_{\lceil \log w \rceil+1}$. Then in time $\lceil \log w \rceil+2$, at most $k$ vertices in $W$ is informed, which is insufficient to finish broadcasting in $G_0$ in time $\lceil \log w\rceil+5$. Thus, $b(G_f, v_0) > \lceil \log w\rceil+5$.

According to the construction of $G_f$, $|V_f|=15w+2^{\lceil \log w \rceil+1}-k$, which is polynomial to $|V_0|=3k+w$.
\end{proof}

\noindent Combining lemmas \ref{a1} and \ref{a2},
\begin{theorem}
\textsc{ST-Broadcast Time} is NP-complete.
\end{theorem}

\bibliographystyle{splncs04}
\bibliography{xjh}

\end{document}